\documentclass[11pt]{article}
\usepackage[utf8]{inputenc}
\usepackage[a4paper]{geometry}

\usepackage{Style/preamble}
\usepackage{Style/en-math}
\usepackage{Style/l-to-ell}

\usepackage{arydshln} 

\DeclareMathOperator{\FB}{FB}
\DeclareMathOperator{\LT}{LT}
\DeclareMathOperator{\Hyp}{Hyp}

\title{Distributed matrix multiplication with straggler tolerance over very small fields}
\author{
  Fidalgo-Díaz, Adrián\\
  \texttt{adrian.fidalgo22@uva.es}
  \and
  Martínez-Peñas, Umberto\\
  \texttt{umberto.martinez@uva.es}
}
\date{}

\begin{document}

\maketitle

\begin{abstract}
    The problem of distributed matrix multiplication with straggler tolerance over finite fields is considered, focusing on field sizes for which previous solutions were not applicable (for instance, the field of two elements). We employ Reed-Muller-type codes for explicitly constructing the desired algorithms and study their parameters by translating the problem into a combinatorial problem involving sums of discrete convex sets. We generalize polynomial codes and matdot codes, discussing the impossibility of the latter being applicable for very small field sizes, while providing optimal solutions for some regimes of parameters in both cases.
\end{abstract}


\section{Introduction}

\subsection{Statement of the problem}
Heavy computations require significant time to execute. One option for improving execution times is to develop better computers with more advanced CPUs. Unfortunately, this solution becomes increasingly challenging each year, as we seem to be approaching a fundamental limit in clock speed. Traditionally, the alternative has been to improve not the infrastructure but the algorithms themselves, for instance, by incorporating parallelization. A parallelizable algorithm divides a task into smaller sub-tasks that can be computed simultaneously. When implementing these algorithms, a typical approach is to distribute the smaller tasks across multiple computers (worker nodes) that operate independently. Once these computations are completed, a master node collects the results and reconstructs the original computation.

Roughly speaking, the greater the number of worker nodes, the greater the speedup. Equivalently, having more worker nodes translates into smaller tasks for each one to compute. However, when implementing these algorithms, real-world problems start to arise. If the number of worker nodes is very high, the expected difference between their execution times becomes significant, say because of the network traffic or due to other reasons. This will induce a bottleneck which limits the performance times of the algorithm since the master node must wait for all the worker nodes to complete their tasks in order to obtain the original computation. This effect is often known as the ``straggler effect'' in the literature. The objetive then is to design algorithms for which the master node can recover the original computation from a subset of worker nodes, considering straggler nodes as non-responsive. Summarizing, we need to recover missing information from the received data, a perfect fit for coding theory methods.

\begin{figure}[h]
    \centering
    \includegraphics[width=8cm]{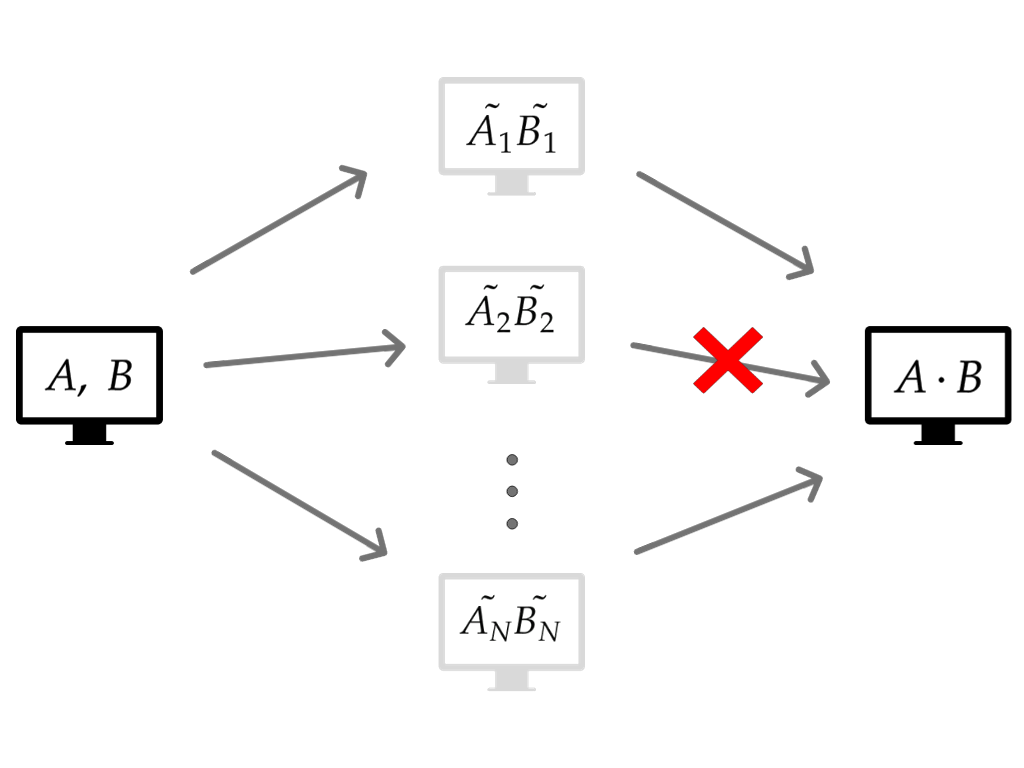}
    \caption{General scheme of DMM with straggler tolerance, where each $\tilde{A}_i$ and $\tilde{B}_i$ denotes a matrix of lower size than $A$ and $B$.}
\end{figure}

In this manuscript we focus on a computation problem that arises in a large variety of problems: multiplying two matrices. More precisely, we consider the scenario where matrices are defined over a finite field. Among the operations which can be reduced to matrix multiplication over finite fields, we briefly mention attacking code-based cryptography \cite{aragon2022bike, pellikaan2017codes, bernstein2008attacking} (or more generally, application to decoding error-correcting codes \cite{pellikaan2017codes}) and solving polynomial equations over finite fields via resultants \cite[Chapter 3]{cox1997ideals}. The first of these applications is particularly notable when $q=2$, since McEliece cryptosystem using binary Goppa codes remains one the most promising cryptosystems of this family.

\subsection{State of the art and our solution}
The most celebrated solution for distributed matrix multiplication with straggler tolerance (DMM, henceforth) was introduced in \cite{polynomial_codes}, where the authors proposed encoding both matrices as evaluations of polynomials before multiplying them, or equivalently, encoding them using Reed-Solomon codes. This approach was later generalized in \cite{polydot_codes}, which also employed Reed-Solomon codes but encoded the computations differently. This solution reduces the number of responsive nodes required to recover the original computation, at the cost of higher per-worker computation and communication costs compared to \cite{polynomial_codes}. So depending on the network specific properties, one of the them will preferred in different scenarios.

Both methods have as downside requiring the number of worker nodes $N$ be smaller than the size of the field over which the matrices are defined. When considering matrices over a finite field of size $q$, the constraint $N \leq q$ is translated into using at most $q$ worker nodes, which is very restrictive especially for small sizes of $q$ like $2$ or $3$. Enlarging the field from $\mathbb{F}_q$ to $\mathbb{F}_{q^\prime}$ for some $q^\prime \geq N$ is not a good solution, as this will result in computational overhead. Even worse, if some worker nodes receive a task that can be computed in $\mathbb{F}_q$, they will complete their tasks before those computing in $\mathbb{F}_{q^\prime}$, therefore aggravating the straggler effect. To address this issue, some alternatives deriving from algebraic geometry codes were proposed \cite{fidalgo2024distributed, machado2023hera, makkonen2023algebraic, matthews2024algebraic, li2024algebraic}, which keep fixed the size of the field while allowing $N \geq q$. For the readers not familiarized with algebraic geometry codes, they work by fixing an algebraic curve over a finite field and evaluating functions on the rational points of the curve in a similar way standard polynomials are evaluated on the elements of the field. For some values of $q$, remarkably $2$ and other non-perfect squares, there exist few options of algebraic curves with many rational points where to evaluate. Consequently, algebraic geometry codes are not an effective solution for these field sizes.

In this work, we propose using a different family of evaluation codes for DMM: codes from multivariate polynomials. More concretely, we propose similar ideas to Reed-Muller codes and hyperbolic codes in order to obtain algorithms for DMM allowing $N \geq q$. Notably, this new approach allows DMM with straggler tolerance over $\mathbb{F}_2$, the field of $2$ elements, something that remained impossible with the previous methods of the state of the art.

\section{Preliminaries. The Footprint Bound} \label{section:preliminaries}

We denote by $ \mathbb{F}_q $ the finite field of size $q$ and define the $\mathbb{F}_q$-algebra
$$ \mathcal{R} = \frac{\mathbb{F}_q[x_1, x_2, \ldots, x_\ell]}{(x_1^q-x_1, \ldots, x_\ell^q-x_\ell)}, $$
where $ (\mathcal{A}) $ denotes the ideal generated by a set $ \mathcal{A} $. Observe that $\mathcal{R}$ is naturally isomorphic to the ring of functions from $\mathbb{F}_q^l$ to $\mathbb{F}_q$ (see Remark \ref{remark:ring_of_functions}). Let $\mathcal{V} \subseteq \mathcal{R} $ be an ($ \mathbb{F}_q $-linear) vector subspace. For a set $\mathbb{P} \subseteq \mathbb{F}_q^\ell $, consider the evaluation map $ev: \mathcal{V} \to \mathbb{F}_q^{\mathbb{P}}$ which sends $f \in \mathcal{V} $ to $\left ( f(P) \right )_{P \in \mathbb{P}}$, its evaluations in $\mathbb{P}$. Observe that $ev$ is well defined and a linear map.

\begin{notation}
    We write vectors in bold.
\end{notation}

For $f \in \mathcal{R} \setminus \{0\}$, we define its restricted \textbf{footprint} as $ \Delta_{q,\prec}(f) := \{ \textbf{a} \in \mathbb{N}^\ell_{<q}\such x^\textbf{a} \notin ( \LT (f) ) \}$, where $ \mathbb{N}_{<q} = \{ 0,1, \ldots, q-1 \} $, $ \prec $ is a monomial ordering (defined in $ \mathbb{N}^\ell $), $ x^\textbf{a} = x_1^{a_1} \cdots x_\ell^{a_\ell} $, for $\textbf{a} \in \mathbb{N}^\ell_{<q} $, and $ \LT(f) $ denotes the leading monomial of $ f $ with respect to $ \prec $. The following lemma can be found in \cite[Sec. IV]{geil2000footprints}.

\begin{lemma}[\cite{geil2000footprints}] \label{lemma:footprint}
    Let $f \in \mathcal{R} \setminus \{0\}$ and $\delta(f) := |\Delta_{q,\prec}(f)| $. If $ Z(f) $  denotes the set of zeros  of $ f $ in $\mathbb{F}_q^l$, then $ |Z(f)| \leq \delta(f) $. In particular, if we define $ \delta(\mathcal{V}) := \max \{\delta(f) \such f \in  \mathcal{V} \setminus \{0\}  \}$, then $|Z(f)| \leq \delta(\mathcal{V}) $, for all $ f \in  \mathcal{V}\setminus \{0\} $.
\end{lemma}

As a consequence, denoting $  k := \delta(\mathcal{V})  $, if $f \in \mathcal{V}$ and $f(P_i) = 0$ for $k + 1$ distinct points $P_1, P_2, \ldots, P_{k+1} \in \mathbb{P}$, then $f = 0$. Equivalently, let $\mathcal{B} := \{f_1, f_2, \ldots, f_{\kappa}\}$ be a basis of $\mathcal{V}$ and $G$ be the matrix associated to $ev$ in such a basis, that is,
\begin{equation}
    \begin{split}
        G :=
        \begin{pmatrix}
            f_1(P_1)            & f_1(P_2)            & \ldots & f_1(P_{k+1}) \\
            f_2(P_1)            & f_2(P_2)            & \ldots & f_2(P_{k+1}) \\
            \vdots              & \vdots              & \ddots & \vdots       \\
            f_{\kappa}(P_{1}) & f_{\kappa}(P_{2}) & \ldots & f_{\kappa}(P_{k+1})
        \end{pmatrix}.
    \end{split}
    \label{equation:matrix_interp}
\end{equation}
Then $G$ has a right inverse, i.e., $ev$ is injective.

\begin{remark}
    In general, the bound on the number of common zeros of $\mathcal{V}$ given in Lemma \ref{lemma:footprint} is not an equality. Nevertheless, if $\mathcal{V}$ is generated by a set of monomials closed under divisibility, then the equality holds (see \cite[Th. 2.8]{camps2020monomial}). Since this is the case for the majority of vector spaces $\mathcal{V}$ that we consider in this manuscript, the bound of Lemma \ref{lemma:footprint} serves as a good proxy for the number of points needed for interpolation.
\end{remark}

\section{Multivariate polynomial codes} \label{section:polynomial}

We present a method for DMM using multivariate polynomials which generalizes polynomial codes \cite{polynomial_codes} and that allows using more than $q$ worker nodes. In fact, the method allows an arbitrary large number of worker nodes for a fixed field size. We give explicit constructions as well as bounds on the recovery threshold (see subsections \ref{subsection:few_variables} and \ref{subsection:many_variables}). Let us discuss the framework.

Denote by $\mathbb{F}_q^{r \times s}$ the $\mathbb{F}_q$-vector space of matrices with $r$ rows and $s$ columns with entries in $\mathbb{F}_q$. Let $A \in \mathbb{F}_q^{r \times s}$ and $B \in \mathbb{F}_q^{s \times t}$ be two matrices. In order to multiply them, we split them as block matrices in the following way:
\begin{equation} \label{equation:matrices_polynomial}
    \begin{split}
        A := \begin{pmatrix}
            A_1 \\
            A_2 \\
            \vdots \\
            A_m \\
        \end{pmatrix}, \quad
        B := \begin{pmatrix}
            B_1 & B_2 & \ldots & B_n
        \end{pmatrix}\\ \implies
        A B =
        \begin{pmatrix}
            A_1 B_1 & A_1 B_2 & \ldots & A_1 B_n\\
            A_2 B_1 & A_2 B_2 & \ldots & A_2 B_n\\
            \vdots & \vdots & \ddots & \vdots\\
            A_m B_1 & A_m B_2 & \ldots & A_m B_n
        \end{pmatrix},
    \end{split}
\end{equation}
where $A_i \in \mathbb{F}_q^{\frac{r}{m} \times s}$ and $B_i \in \mathbb{F}_q^{s \times \frac{t}{n}}$. We want to find $p_A \in \mathcal{R}^{\frac{r}{m} \times s}$ and $p_B \in \mathcal{R}^{s \times \frac{t}{n}}$ such that
\begin{equation*}
    p_A := \sum_{i=1}^m A_i x^{\textbf{a}_i}, \quad p_B := \sum_{j=1}^n B_j x^{\textbf{b}_j},
\end{equation*}
and such that, for every $\textbf{a}, \textbf{a}^\prime \in D_A := \{ \textbf{a}_1, \ldots, \textbf{a}_m \}$ and $\textbf{b}, \textbf{b}^\prime \in D_B := \{ \textbf{b}_1, \ldots, \textbf{b}_n \}$,
\begin{equation} \label{equation:non_colliding}
    (\textbf{a} + \textbf{b})_q = (\textbf{a}^\prime + \textbf{b}^\prime)_q \implies (\textbf{a}, \textbf{b}) = (\textbf{a}^\prime, \textbf{b}^\prime),
\end{equation}
where, given $a \in \mathbb{N}_{<2q}$ we define $(a)_q$ as
\begin{equation*}
    (a)_q :=
    \begin{cases}
        a & \text{if } a < q,\\
        (a \mod q) + 1  & \text{if } q \leq a < 2q,
    \end{cases}
\end{equation*}
and we extend it coordinatewise to $\textbf{a} \in \mathbb{N}^l$.
\begin{remark} \label{remark:ring_of_functions}
    Observe that $\mathcal{R}$ is, by definition, the ring of polynomials quotient the following equivalence relation: $f \sim g$ if and only if they define the same function when evaluated in $\mathbb{F}_q^l$. Moreover, the operation $(\cdot)_q$ is the reduction of the exponent of a monomial to its cannonical representation, i.e., $x^\textbf{a} x^\textbf{b} = x^{\textbf{a}+\textbf{b}} = x^{(\textbf{a}+\textbf{b})_q} \in \mathcal{R}$ for all $\textbf{a},\textbf{b} \in \mathbb{N}_{<q}^\ell $.
\end{remark}
Next, set $h := p_A p_B \in \mathcal{R}^{\frac{r}{m} \times \frac{t}{n}}$.
\begin{remark} \label{remark:non_colliding}
    By enforcing (\ref{equation:non_colliding}) we ensure that $A_i B_j$ is the coefficient of $x^{\textbf{a}_i + \textbf{b}_j}$ in $h$, so we can retrieve $AB$ from $h$.
\end{remark}
Observe that if $h_{i,j} \in \mathcal{R} $ is the $(i,j)$th coordinate of $h$, then $h_{i,j} \in \mathcal{V} := \langle x^{\textbf{a} + \textbf{b}} \such \textbf{a} \in D_A, \,\textbf{b} \in D_B \rangle_{\mathbb{F}_q} $, where $ \langle \cdot \rangle_{\mathbb{F}_q} $ denotes $ \mathbb{F}_q $-linear span. So from Section \ref{section:preliminaries}, we can interpolate $h$ from knowing any $k+1$ evaluations (the image of $h$ by $ev$ defined with $|\mathbb{P}| = k+1$), where $ k = \delta(\mathcal{V}) $, by inverting the matrix $ G $ from (\ref{equation:matrix_interp}) on the right. This interpolation step has negligible computational complexity compared to that of the matrix multiplication by each worker when the matrix sizes $ r $, $ s $ and $ t $ are large, as in \cite[Sec. 5]{fidalgo2024distributed} (see Appendix \ref{appendix:complexity}).

This results in the following \textbf{algorithm}: the master node shares the evaluations $p_A(P_i) \in \mathbb{F}_q^{\frac{r}{m} \times s}$ and $p_B(P_i) \in \mathbb{F}_q^{s \times \frac{t}{n}}$, with the i-th worker. Then, the workers compute the matrix multiplications $h(P_i) = p_A (P_i) p_B(P_i)$ in parallel. When any $k+1$ of them end their computations, the master node recovers $h$ by performing componentwise interpolations, and so recovering $AB$ because of Remark \ref{remark:non_colliding}.  In the literature, the number $k + 1$ is called the \textbf{recovery threshold}, i.e., the minimum number of responsive nodes needed to recover the product $AB$.

The remaining part is how to select sets $ D_A, D_B \subseteq \mathbb{N}_{<q}^\ell$ of sizes $|D_A| = m$ and $|D_B| = n$ satisfying (\ref{equation:non_colliding}) in such way that $k$ is as small as possible. Noticing that $\mathcal{V}$ is generated by monomials, we conclude that
\begin{equation} \label{equation:optimal_k}
    \begin{split}
        k = \delta(\mathcal{V}) = \max \left \{ |\Delta(x^\textbf{a})| \such x^\textbf{a} \in \mathcal{V} \right \} = q^\ell - \min \left \{\prod_{i=1}^{\ell} (q - (a_i)_q) \such x^{\textbf{a}} \in \mathcal V \right \}.
    \end{split}
\end{equation}
We summarize the problem of obtaining the coefficients of $p_A$ and $p_B$ satisfying (\ref{equation:non_colliding}) and minimizing (\ref{equation:optimal_k}) in Problem \ref{problem:polynomial}.

\begin{notation} \label{notation:minkowski}
    We define $D_A +_q D_B$ as the Minkowski sum reduced with $(\cdot)_q$, that is
    \begin{equation*}
        D_A +_q D_B := \{(\textbf{a} + \textbf{b})_q \in \mathbb{N}^l_{<q} \such \textbf{a} \in D_A, \textbf{b} \in D_B\}.
    \end{equation*}
\end{notation}

\begin{problem} \label{problem:polynomial}
    Find $D_A, D_B \subseteq \mathbb{N}^\ell_{<q}$ such that
    \begin{enumerate}
        \item $|D_A| = m$ and $|D_B| = n$,
        \item $|D_A +_q D_B| = |D_A| \cdot |D_B|$ or, equivalently, satisfying (\ref{equation:non_colliding}), and
        \item $\FB(D_A +_q D_B) := \min \{\prod_{i=1}^\ell (q - (a_i + b_i)_q) \such \textbf{a} \in D_A, \textbf{b} \in D_B\}$ is as large as possible.
    \end{enumerate}
    If $D_A$ and $D_B$ achieve items 1 and 2, we say they are a solution. If item 3 is also achieved, we say they are an optimal solution.
\end{problem}

More explicitly, finding a solution for Problem \ref{problem:polynomial} will yield an algorithm for DMM with straggler tolerance for matrices which are subdivided in $m$ and $n$ submatrices as in (\ref{equation:matrices_polynomial}), with recovery threshold $q^l - \FB(D_A +_q D_B) + 1$, as noted in (\ref{equation:optimal_k}). So minimizing the recovery threshold is equivalent to maximizing $\FB(D_A +_q D_B)$. We study such solutions specifying both $\FB(D_A +_q D_B)$ and the recovery threshold for clarity, but the reader should keep in mind that one uniquely determines the other. Even though in Section \ref{section:preliminaries} we defined the footprint such that $\delta(f) = q^l - \FB(\LT(f))$, often we will refer to $\FB(D_A + D_B)$ simply as the footprint during the rest of the manuscript.

\begin{remark}
    Notice that the maximum number of worker nodes is $q^l$, the maximum number of different evaluations we can perform over $\mathbb{F}_q^l$. In the case $ \ell = 1$ (polynomials defined over one variable), this algorithm is known as polynomial codes \cite{polynomial_codes} in the literature, and Problem \ref{problem:polynomial} is solved by choosing $D_A := \{0, 1, \ldots, m-1\}$ and $D_B := \{0, m, \ldots, (n-1)m\}$, always under the assumption that $mn < q$. Polynomial codes attain the best information theoretical recovery threshold but the number of worker nodes is limited by the size of the field the matrices are defined over. In particular, small fields such as $ \mathbb{F}_2 $ or $ \mathbb{F}_3 $ are not allowed by polynomial codes \cite{polynomial_codes}. Neither constructions arising from algebraic geometry codes \cite{fidalgo2024distributed, matthews2024algebraic, makkonen2023algebraic, li2024algebraic} are applicable for these field sizes because of the lack of algebraic function fields over them.
\end{remark}

\begin{remark}
    One of the first things we notice when looking for optimal solutions to Problem \ref{problem:polynomial} is that if $D_A$ (analogously $D_B$) is not ``laying on the axes'', then the solution $D_A$ and $D_B$ is not optimal. More formally, if $d_i := \min \{a_i \such \textbf{a} \in D_A\} > 0$, we can consider $D_A^\prime := D_A - (0,\ldots,d_i,\ldots,0)$ together with $D_B$, obtaining a strictly better solution in terms of the footprint bound.
\end{remark}

\subsection{A bound on the recovery threshold}

We start by introducing a bound on $\FB(D_A +_q D_B)$ for a solution to Problem \ref{problem:polynomial} and, consequently, a bound on the recovery threshold the algorithm produces. The next definition is \cite[Def. 4]{garcia2020high}.

\begin{definition}[\cite{garcia2020high}] \label{definition:hyperbolic}
    Let $q, F, l \in \mathbb{N}$, we define the hyperbolic set $\Hyp_q(F,\ell)$ or simply $\Hyp(F)$ as
    \begin{equation*}
        \Hyp_q(F, \ell) := \left \{\textbf{a} \in \mathbb{N}_{<q}^\ell \such \prod_{i=1}^l (q - a_i) \geq F \right \}.
    \end{equation*}
\end{definition}

\begin{figure}
    \centering
    \resizebox{!}{200pt}{
        \begin{tikzpicture}
	\draw[thick,->] (0,0) -- (10.6, 0) node {};
	\draw[thick,->] (0,0) -- (0, 10.6) node {};
	\draw (0 cm, 0pt) node[anchor=north east] {$0$};
	\foreach \x in {5,10,...,10}
	    \draw (\x cm, -1pt) node[anchor=north] {$\x$};
	\foreach \y in {5,10,...,10}
	    \draw (-1pt, \y cm) node[anchor=east] {$\y$};

	\foreach \x in {0,...,10}
	    \foreach \y in {0,...,10}
	        \fill[color=gray] (\x, \y) circle[radius=1pt] node {};

	\fill[fill= blue, opacity=0.5](0, 10) -- (0, 0) -- (10, 0) -- (10, 3) -- (9, 7) -- (7, 9) -- (3, 10) -- cycle; 
	\fill[fill= red, opacity=0.5](0, 0) -- (8, 0) -- (8, 3) -- (7, 5) -- (5, 7) -- (3, 8) -- (0, 8) -- cycle; 
	\fill[fill= yellow, opacity=0.5](0, 0) -- (6, 0) -- (5, 2) -- (2, 5) -- (0, 6) -- cycle; 
\end{tikzpicture}
    }
    \caption{Hyperbolic sets $\Hyp_{11}(53,2) \subseteq \Hyp_{11}(24,2) \subseteq \Hyp_{11}(8,2)$ represented in yellow, red and blue, respectively.}
\end{figure}
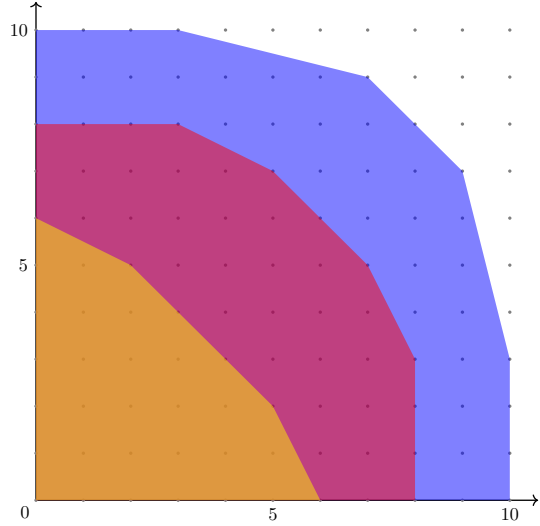

The hyperbolic set $\Hyp(F)$ is, by definition, the biggest subset of $\mathbb{N}^l_{<q}$ such that $\FB(\Hyp(F)) \geq F$. Proposition \ref{proposition:polynomial_bound} uses this concept to give a bound on $\FB(D_A +_q D_B)$.

\begin{proposition} \label{proposition:polynomial_bound}
    Consider $\xi := \max \{F \in \mathbb{N} \such |\Hyp(F)| \geq mn\}$. If $D_A$ and $D_B$ is a solution, then $\FB(D_A +_q D_B) \leq \xi$.
\end{proposition}

\begin{proof}
    Due to the definition of $\FB(D_A +_q D_B)$ we have that $\prod_{i=1}^\ell (q - c_i) \geq \FB(D_A +_q D_B)$ for every $\textbf{c} \in D_A +_q D_B$. Then, $D_A +_q D_B \subseteq \Hyp(\FB(D_A +_q D_B))$. Comparing their sizes we get that $mn \leq |\Hyp(\FB(D_A +_q D_B))|$, and the result follows from the definition of $\xi$.
\end{proof}

Despite the fact that the bound of Proposition \ref{proposition:polynomial_bound} is not a closed formula, we will see in Corollary \ref{corollary:recursive_hyperbolic} that we can easily compute it. Because of this, in the Appendix we will be able to compare the footprints of our constructions with this bound.

Now we introduce different solutions to Problem \ref{problem:polynomial}. We construct the solutions and study the parameters in general, but we group them according to on which regime of the parameters $q$ and $l$ they yield better results. Observe that there is a tradeoff between these two parameters: more variables lead to worse solutions but they are needed when working over fields of small size. This is, either $q$ or $l$ have to be sufficiently large.

\subsection{Few variables and moderately small field size} \label{subsection:few_variables}

We start studying constructions to solve Problem \ref{problem:polynomial} that behave well for small $l$ and moderate $q$. As a naive generalization of the chosen degrees in polynomial codes \cite{polynomial_codes}, we introduce Construction \ref{construction:box_poly}.

\begin{notation}
    Let $\textbf{s}, \textbf{k} \in \mathbb{N}^l$, we write $\textbf{k} \ast \textbf{s} := (k_1 s_1, k_2 s_2, \ldots, k_l s_l)$. If $S \subseteq \mathbb{N}^l$, then we write $\textbf{k} \ast S := \{ \textbf{k} \ast \textbf{s} \in \mathbb{N}^l \such \textbf{s} \in S\}$. This operation, sometimes denoted by $\dot{\ast}$, is often referred to as the Schur product or the star product \cite{randriambololona2015products}.
\end{notation}

\begin{construction}[Box] \label{construction:box_poly}
    Let $\textbf{m}, \textbf{n} \in \mathbb{N}^\ell_{< q}$ such that $m_i n_i \leq q$ for every $i = 1, 2, \ldots, l$. Define
    \begin{equation*}
        \begin{split}
            &D_A := \{\textbf{a} \in \mathbb{N}^l_{<q} \such a_i < m_i\},\\
            &D_B := \{\textbf{m} \ast \textbf{k} \in \mathbb{N}^l_{<q} \such k_i < n_i\}.
        \end{split}
    \end{equation*}
    Because of $m_i n_i \leq q$, we have that $\textbf{a} + \textbf{b} \in \mathbb{N}^l_{<q}$ for every $\textbf{a} \in D_A$ and $\textbf{b} \in D_B$, resulting in $D_A +_q D_B = D_A + D_B$ being the classical Minkowski sum. By Euclidean division by $m_i$, we deduce that $|D_A + D_B| = |D_A| |D_B|$, and so $D_A$ and $D_B$ are a solution to Problem \ref{problem:polynomial}. Moreover, $|D_A| = \prod_{i=1}^\ell m_i$, $|D_B| = \prod_{i=1}^\ell n_i$ and
    \begin{equation*}
        \FB(D_A + D_B) = \prod_{i=1}^\ell \left (q - m_i + 1 - m_i(n_i - 1) \right ) = \prod_{i=1}^\ell (q - m_i n_i + 1),
    \end{equation*}
    so the recovery threshold is $q^l - \prod_{i=1}^\ell (q - m_i n_i + 1) + 1$.
\end{construction}

The idea behind Construction \ref{construction:box_poly} is to pick $D_A$ and $D_B^\prime$ as two hypercubes (boxes) and obtain $D_B$ by expanding $D_B^\prime$ with the coordinatewise multiplication (star product), ensuring that $|D_A + D_B| = |D_A| |D_B|$. This construction can be generalized selecting $D_A$ and $D_B^\prime$ not necessarily as boxes, and obtaining $D_B$ with the corresponding expansion as Proposition \ref{proposition:expand_DB} shows.

\begin{proposition} \label{proposition:expand_DB}
    Let $D_A, D_B^\prime \subseteq \mathbb{N}_{<q}^l$ and $m_i := 1 + \max_{\textbf{a}, \textbf{a}^\prime \in D_A} |a_i - a_i^\prime|$. Consider $D_B := \textbf{m} \ast D_B^\prime$, where $\textbf{m} := (m_1, m_2, \ldots, m_l)$. If $\max_{(\textbf{a}, \textbf{b}) \in D_A \times D_B} (a_i + b_i) < q$ for every $i = 1, 2, \ldots, l$, then $|D_A +_q D_B| = |D_A| |D_B|$, thus $D_A$ and $D_B$ are a solution to Problem \ref{problem:polynomial}.
\end{proposition}

\begin{proof}
    Let $(\textbf{a}, \textbf{b}), (\textbf{a}^\prime, \textbf{b}^\prime) \in D_A \times D_B$ such that $\textbf{a} + \textbf{b} = \textbf{a}^\prime + \textbf{b}^\prime$, then $|a_i - a_i^\prime| = |b_i - b_i^\prime|$ for every $i = 1, 2, \ldots, l$. If $(\textbf{a}, \textbf{b}) \neq (\textbf{a}^\prime, \textbf{b}^\prime)$, then $\textbf{a} \neq \textbf{a}^\prime$ and $\textbf{b} \neq \textbf{b}^\prime$ and
    \begin{equation*}
        |a_i - a^\prime_i| < m_i \leq |b_i - b^\prime_i|,
    \end{equation*}
    for some $i$. This contradiction yields $(\textbf{a}, \textbf{b}) = (\textbf{a}^\prime, \textbf{b}^\prime)$, so $|D_A +_q D_B| = |D_A| |D_B|$.
\end{proof}

The approach of Proposition \ref{proposition:expand_DB} is expanding $D_B^\prime$ by the limits of the smallest box containing $D_A$. We can obtain a ``better'' solution by simply selecting $D_A$ as the whole box delimited by $\textbf{m}$. In fact, we can improve this construction as Construction \ref{construction:better_box} shows, choosing $D_B^\prime$ as large as it can be.

\begin{construction}[Better box] \label{construction:better_box}
    Let $\textbf{m} \in \mathbb{N}_{<q}^l$ and $F \in \mathbb{N}$. Consider the sets
    \begin{equation*}
        \begin{split}
            &D_A := \{\textbf{a} \in \mathbb{N}^l_{<q} \such a_i < m_i\},\\
            &D_B := \left \{ \textbf{m} \ast \textbf{b} \in \mathbb{N}_{<q}^l \such \prod_{i=1}^l (q - m_i + 1 - m_i b_i) \geq F \right \}.
            \end{split}
    \end{equation*}
    Proposition \ref{proposition:expand_DB} ensures $D_A$ and $D_B$ to be a solution to Problem \ref{problem:polynomial}. Moreover, by the definition of $D_B$, we have $\FB(D_A + D_B) \geq F$, so the recovery threshold is $k + 1 \leq q^l - F + 1$.
\end{construction}

Construction \ref{construction:better_box} is at least as good as (in the majority of cases is strictly better than) Construction \ref{construction:box_poly}, since one can always set $F$ in Construction \ref{construction:better_box} to be $\FB(D_A + D_B)$ with $D_A$ and $D_B$ as in Construction \ref{construction:box_poly}, obtaining a better box solution with at least the same footprint and at least the same sizes as sets $D_A$ and $D_B$. The only non explicit parameter of Construction \ref{construction:better_box} is the size of $D_B$. Proposition \ref{proposition:recursive_better_box} gives a useful recurrence for computing it.

\begin{definition}
    Let $\textbf{m} \in \mathbb{N}_{<q}^l$ and $F \in \mathbb{N}$. Define $q_i := q - m_i + 1$ and
    \begin{equation*}
        D_B(F, l) := \left \{ (m_1 b_1, m_2 b_2, \ldots, m_l b_l) \in \mathbb{N}_{<q}^l \such \prod_{i=1}^{l}(q_i - m_i b_i) \geq F \right \}.
    \end{equation*}
    In particular, $D_B = D_B(F,l)$ in Construction \ref{construction:better_box}.
\end{definition}

\begin{proposition} \label{proposition:recursive_better_box}
    Let $\textbf{m} \in \mathbb{N}_{<q}^l$ and $F \in \mathbb{N}$. Then,
    \begin{equation*}
        \begin{split}
        |D_B(F,l)| =
            \begin{cases}
            \max \{0, \lfloor \frac{q_1 - F}{m_1} \rfloor + 1\} &\text{if } l = 1,\\
            \sum_{b=0}^{\lfloor \frac{q_l-1}{m_l}\rfloor} |D_B(\lceil \frac{F}{q_l-m_l b}\rceil, l-1)| &\text{if } l > 1.
            \end{cases}
        \end{split}
    \end{equation*}
\end{proposition}

\begin{proof}
    If $l = 1$, then
    \begin{equation*}
        \begin{split}
            |D_B(F,1)| &= |\{ m_1 b_1 \in \mathbb{N}_{<q} \such q_1 - m_1 b_1 \geq F\}|\\
            &= \left | \left \{ m_1 b_1 \in \mathbb{N}_{<q} \such \left \lfloor \frac{q_1 - F}{m_1} \right \rfloor \geq b_1 \right \} \right |\\
            &= \max \left \{0, \left \lfloor \frac{q_1 - F}{m_1} \right \rfloor + 1 \right\}.
        \end{split}
    \end{equation*}
    If $l > 1$, we partition $D_B(F,l)$ and the sum of the sizes gives us the formula:
    \begin{equation*}
        \begin{split}
            |D_B(F,l)| &= \left |\bigcup_{b=0}^{\lfloor \frac{q_l-1}{m_l} \rfloor} \left \{ (m_1 b_1, \ldots, m_{l-1} b_{l-1}, m_l b) \in \mathbb{N}^{l}_{<q} \such \prod_{i=1}^{l-1}(q_i - m_i b_i) \geq  \left \lceil \frac{F}{q_l-m_l b} \right \rceil \right \} \right| \\
            &= \sum_{b=0}^{\lfloor \frac{q_l-1}{m_l} \rfloor} \left | \left \{(m_1 b_1, \ldots, m_{l-1} b_{l-1}) \in \mathbb{N}_{<q}^{l-1} \such \prod_{i=1}^{l-1}(q_i - m_i b_i) \geq \left \lceil \frac{F}{q_l-m_l b} \right \rceil \right \} \right |\\
            &= \sum_{b=0}^{\lfloor \frac{q_l-1}{m_l}\rfloor} \left |D_B \left (\left \lceil \frac{F}{q_l-m_l b} \right \rceil, l-1 \right ) \right |.
        \end{split}
    \end{equation*}
\end{proof}

In Proposition \ref{proposition:recursive_better_box}, the definition of $D_B(F,l)$ recovers that of $\Hyp_q(F,l)$, in the sense that $D_B(F,l) = \Hyp(F,l)$ by making $\textbf{m} = (1,1, \ldots, 1)$. Moreover, observe that the proof can be adapted to compute the set $D_B(F,l)$ itself, not only its size, by using a backtracking-like algorithm, for example.

As a corollary of Proposition \ref{proposition:recursive_better_box} we derive the result from \cite{geil2001hyperbolic} that computes the size of hyperbolic sets. The motivation of this result is to compute the bound on the recovery threshold of Proposition \ref{proposition:polynomial_bound}. We state it in Corollary \ref{corollary:recursive_hyperbolic}.

\begin{corollary}[\cite{geil2001hyperbolic}] \label{corollary:recursive_hyperbolic}
    Let $F \geq 1$. Then
    \begin{equation*}
        \begin{split}
            |\Hyp_q(F,\ell)| =
            \begin{cases}
                \max \{ 0, q-F+1 \} &\text{if } \ell = 1\\
                \sum_{i=1}^q |\Hyp_q(\lceil \frac{F}{i} \rceil, \ell - 1)| &\text{if } \ell > 1.
            \end{cases}
        \end{split}
    \end{equation*}
\end{corollary}

\begin{proof}
    By setting $i=q-b$ and $\textbf{m} = (1,1,\ldots,1)$, the results follows as a particular case of Proposition \ref{proposition:recursive_better_box}.
\end{proof}

When restricting to $l = 2$, only two variables, observe that we obtain simpler formulas for computing the size of both the set $D_B$ in Construction \ref{construction:better_box} and the set $\Hyp_q(F)$.

\begin{corollary}
    Let $l=2$ and $\textbf{m} \in \mathbb{N}_{<q}^2$. Then
    \begin{equation*}
        |D_B(F,2)| = \sum_{b=0}^{\lfloor \frac{q_2-1}{m_2}\rfloor} \max \left \{0, \left \lfloor \frac{q_1 - \left \lceil \frac{F}{q_2 - m_2b} \right \rceil}{m_1} \right \rfloor + 1 \right \}.
    \end{equation*}
    Regarding hyperbolic sets, for $F \geq 1$ we have
    \begin{equation*}
        |\Hyp_q(F,2)| = \sum_{k=1}^q \max \left \{0 , q - \left \lceil \frac{F}{k} \right \rceil + 1 \right \}.
    \end{equation*}
\end{corollary}

\begin{proof}
    Immediate from $l=2$ in Proposition \ref{proposition:recursive_better_box} and Corollary $\ref{corollary:recursive_hyperbolic}$, respectively.
\end{proof}

Tables \ref{table:poly_box1} and \ref{table:poly_box2} in the Appendix contain the parameters for several examples of Constructions \ref{construction:box_poly} and \ref{construction:better_box} when defined over two variables together with the bound given by Proposition \ref{proposition:polynomial_bound}.

\subsection{Many variables and small field size} \label{subsection:many_variables}

Now we give a solution to Problem \ref{problem:polynomial} that works well when multiplying matrices over a very small field, that is very small $q$. We state this solution in Construction \ref{construction:separation_of_variables}.

\begin{construction}[Separation of variables] \label{construction:separation_of_variables}
    Let $\ell = m^\prime + n^\prime$, $0 \leq F_A \leq m^\prime$ and $0 \leq F_B \leq n^\prime$. Define
    \begin{equation*}
        \begin{split}
            &D_A := \{(a_1, a_2, \ldots, a_{m^\prime}, 0, \ldots, 0) \in \mathbb{N}^\ell_{< q} \such \textbf{a} \in \Hyp_q(F_A, m^\prime)\},\\
            &D_B := \{(0, \ldots, 0, b_{m^\prime + 1}, b_{m^\prime + 2}, \ldots, b_{\ell}) \in \mathbb{N}^\ell_{< q} \such \textbf{b} \in \Hyp_q(F_B, n^\prime) \}.
        \end{split}
    \end{equation*}
    We have that $|D_A + D_B| = |D_A| |D_B|$ and, from Definition \ref{definition:hyperbolic}, we obtain the bound
    \begin{equation*}
        \FB(D_A + D_B) = \min \left \{\prod_{i=1}^{m^\prime} (q - a_i) \such \textbf{a} \in D_A \right \} \min \left \{\prod_{i=1}^{n^\prime} (q - b_{m^\prime + i}) \such \textbf{b} \in D_B \right \} \geq F_A F_B,
    \end{equation*}
    yielding a recovery threshold of $k + 1\leq q^l - F_A F_B + 1$. Notice that $|D_A| = |\Hyp_q(F_A, m^\prime)|$ and $|D_B| = |\Hyp_q(F_B, n^\prime)|$. This is the reason for selecting $D_A$ and $D_B$ from hyperbolic sets, because, by Definition \ref{definition:hyperbolic}, this maximizes their sizes when using this method of separation of variables.
\end{construction}

\begin{remark} \label{remark:box_f2}
    Construction \ref{construction:separation_of_variables} may seem naive, but in general it is the best option when considering very small values of $q$, like $2, 3$ or $5$. For these $q$, Construction \ref{construction:better_box} (and also Construction \ref{construction:box_poly}) has very restricted parameters. For example, for $q=2$, the box $\mathbb{N}_{<q}^l$ has only two elements in each coordinate, so in Construction \ref{construction:better_box}, $\textbf{m}$ is a vector consisting of $0$s and $1$s. If $m_i = 0$ for some $i$, then $D_A$ is empty. So $\textbf{m}$ has to be $(1,1,\ldots,1)$, but then $D_A$ consists of only one element, $(0,0,\ldots, 0)$. We conclude that, for $q=2$, we are restricted in Construction \ref{construction:better_box} to $m=1$, that is, performing DMM only by splitting matrix $B$ and not $A$, in other words, ``introducing redundancy only in one of the matrices''. This strategy of encoding only one of the operands was one of the first ideas of the coded computation schemes \cite{dutta2017coded}, and presents the disadvantage of performing bad in terms of the communication cost. Observe as well that for $q=2$ and $l=2$, we have $|\mathbb{N}_{<2}^2| = 4$ and so $mn \leq 4$, making evident the necessity of using a large number of variables when the field size is small.
\end{remark}

Similar to Construction \ref{construction:better_box}, in Construction \ref{construction:separation_of_variables} the sizes of $D_A$ and $D_B$ are not explicit. Fortunately, this is not a problem when giving examples since we can efficiently compute them by using the recurrence of Corollary \ref{corollary:recursive_hyperbolic}.

To conclude this section, we focus on the case of matrices defined over $\mathbb{F}_2$. In this extremal scenario, Constructions \ref{construction:box_poly} and \ref{construction:better_box} are not applicable by Remark \ref{remark:box_f2}. In general for $q=2$, $mn \leq 2^l$, or equivalently $l \geq \log_2(mn)$, is required. As far as the authors know, there are currently no general methods for performing DMM with straggler tolerance over $\mathbb{F}_2$, apart from the one presented in this subsection.

\begin{proposition} \label{proposition:hyp2}
    Let $F, l \in \mathbb{N}$. Then
    \begin{equation*}
        |\Hyp_2(F,l)| = \sum_{i=0}^{l - \lceil \log_2 F \rceil } \binom{l}{i}.
    \end{equation*}
\end{proposition}

\begin{proof}
    Let $\textbf{a} \in \mathbb{N}_{<q}^l$, we denote by $\wt(\textbf{a})$ the Hamming weight of $\textbf{a}$, i.e., the number of nonzero coordinates of $\textbf{a}$. When $q = 2$, we observe that given $\textbf{a} \in \mathbb{N}_{<2}^l = \{0,1\}^l$, we have that $\prod_{i=1}^l (2 - a_i) = 2^{l - \wt (\textbf{a})}$. Hence
    \begin{equation*}
    \begin{split}
        |\Hyp_2(F,l)| &= |\{\textbf{a} \in \{0,1\}^l \such 2^{l - \wt(\textbf{a})} \geq F\}|\\
        &= |\{\textbf{a} \in \{0,1\}^l \such \wt(\textbf{a}) \leq l - \lceil \log_2 F \rceil\}| = \sum_{i=0}^{l - \lceil \log_2 F \rceil } \binom{l}{i}.
    \end{split}
    \end{equation*}
\end{proof}

\begin{remark}
    Observe that, when $q=2$, the evaluation code defined by a hyperbolic set is a Reed-Muller code. So separation of variables (Construction \ref{construction:separation_of_variables}) in $\mathbb{F}_2$ proposes encoding the matrices $A$ and $B$ using two Reed-Muller codes.
\end{remark}

Proposition \ref{proposition:hyp2} gives us a way for fast computing the sizes of hyperbolic sets for $q=2$, and so the sizes of $D_A$ and $D_B$ in Construction \ref{construction:separation_of_variables}. In addition, it shows us which are the greatest designed footprints such that $|D_A|$ and $|D_B|$ are maximized: we have to pick $F_1 = 2^\alpha$ and $F_2 = 2^\beta$ for some $\alpha, \beta \in \mathbb{N}$. This motivates the following corollary.

\begin{corollary}
    Let $D_A$ and $D_B$ as in Construction \ref{construction:separation_of_variables} for $q=2$. If $F_A = 2^\alpha$ and $F_B = 2^\beta$, then
    \begin{equation*}
        |D_A| = \sum_{i=0}^{m^\prime - \alpha} \binom{m^\prime}{i} \quad \text{and} \quad |D_B| = \sum_{i=0}^{n^\prime - \beta} \binom{n^\prime}{i}.
    \end{equation*}
\end{corollary}

\begin{proof}
    Straightforward from Proposition \ref{proposition:hyp2}.
\end{proof}

%
Tables \ref{table:poly_varsep1} and \ref{table:poly_varsep2} in the Appendix contain the parameters for several examples of Construction \ref{construction:separation_of_variables} when $q=2$ together with the bound given by Proposition \ref{proposition:polynomial_bound}. Tables \ref{table:poly_varsep3} and \ref{table:poly_varsep4} contain the parameters for $q>2$. Observe that, when $q=2$, Construction \ref{construction:separation_of_variables} is optimal for a large number of non-trivial values of $m$.

\section{Multivariate matdot codes}

Now we follow the trail of matdot codes \cite{polydot_codes}, a different method for performing DMM. Matdot codes achieve lower recovery thresholds than polynomial codes \cite{polynomial_codes} at the expense of having a higher communication and computational cost. With the same ideas, in this section we propose polynomials in several variables for DMM as the algorithm presented in Section \ref{section:polynomial} does, but splitting the matrices in a different way. This yields a different algorithm with lower recovery threshold but higher per-woker computation cost (see Appendix \ref{appendix:complexity}).

Let $A \in \mathbb{F}_q^{r \times s}$ and $B \in \mathbb{F}_q^{s \times t}$, consider the subdivision in blocks of matrices given by
\begin{equation}
    \begin{split}
        A =
        \begin{pmatrix}
            A_1 & A_2 & \ldots & A_m
        \end{pmatrix}, \quad
        B = \begin{pmatrix}
            B_1\\
            B_2\\
            \vdots\\
            B_m\\
        \end{pmatrix}\\ \implies
        A B = A_1 B_1 + A_2 B_2 + \cdots + A_m B_m,
    \end{split}
\end{equation}
where $A_i \in \mathbb{F}_q^{r \times \frac{s}{m}}$ and $B_i \in \mathbb{F}_q^{\frac{s}{m} \times t}$. As in Section \ref{section:polynomial}, we define $p_A \in \mathcal{R}^{ r \times \frac{s}{m}}$ and $p_B \in \mathcal{R}^{\frac{s}{m} \times t}$ such that
\begin{equation*}
    p_A := \sum_{i=1}^m A_i x^{\textbf{a}_i}, \quad p_B := \sum_{i=1}^m B_i x^{\textbf{b}_i},
\end{equation*}
but, this time satisfying that, there exist exactly $m$ distinct pairs $(\textbf{a},\textbf{b}) \in D_A \times D_B$, where $D_A := \{ \textbf{a}_1, \ldots, \textbf{a}_m \} $ and $D_B := \{ \textbf{b}_1, \ldots, \textbf{b}_m \}$, such that $\textbf{d} = (\textbf{a} + \textbf{b})_q$ for a fixed $\textbf{d} \in \mathbb{N}_{<q}^l$. Here, $(\cdot)_q$ denotes the natural sum of exponents of monomials in $\mathcal{R}$, as in Section \ref{section:polynomial}. The property is not arbitrary, since it implies $AB$ to be the coefficient of the monomial $x^\textbf{d}$ in $h := p_A p_B$. Using this fact we can design the following \textbf{algorithm}.

First, the master node shares $p_A(P_i) \in \mathbb{F}_q^{r \times \frac{s}{m}}$ and $p_B(P_i) \in \mathbb{F}_q^{\frac{s}{m} \times t}$ with the $i$-th worker node. Then, the worker nodes compute the matrix multiplications $h(P_i) = p_A(P_i) p_B(P_i)$ and give back the result. When enough worker nodes have responded, the master node uses the evaluations $h(P_i)$ to recover $h$ and, consequently, to obtain $AB$.

The amount of responsive worker nodes necessary to recover $h$ depends on the choice of the sets $D_A$ and $D_B$, as (\ref{equation:optimal_k}) in Section \ref{section:polynomial} summarizes. The objective then is to minimize $k + 1 := q^l - \min_{(\textbf{a}, \textbf{b}) \in D_A \times D_B} \{\prod_{i=1}^\ell (q - (a_i + b_i)_q)\} + 1$, the number of evaluations that ensures we can interpolate $h$. We express this in Problem \ref{problem:matdot} in terms of the sum defined in Notation \ref{notation:minkowski}.

\begin{problem} \label{problem:matdot}
    Find $D_A, D_B \subseteq \mathbb{N}^\ell_{<q}$ such that
    \begin{enumerate}
        \item $|D_A| = |D_B| = m$.
        \item There exists $\textbf{d} \in \mathbb{N}_{<q}^l$, such that there are exactly $m$ pairs $(\textbf{a}_i, \textbf{b}_i) \in D_A \times D_B$ for $i=1,2,\ldots,m$ such that $\textbf{d} = \textbf{a}_i + \textbf{b}_i$ and satisfying that $\textbf{a}_i \neq \textbf{a}_j$ and $\textbf{b}_i \neq \textbf{b}_j$ if $i \neq j$.
        \item $\FB(D_A +_q D_B) := \min \{\prod_{i=1}^\ell (q - (a_i + b_i)_q) \such \textbf{a} \in D_A, \textbf{b} \in D_B\}$ is as large as possible.
    \end{enumerate}
    If $D_A$ and $D_B$ achieve items 1 and 2, we say they are a solution. If item 3 is also achieved, we say they are an optimal solution. Sometimes, we will refer to $D_A$, $D_B$ and $\textbf{d}$ as the solution itself, even though $\textbf{d}$ is dependent on the sets $D_A$ and $D_B$.
\end{problem}

\begin{remark} \label{remark:useless_variables}
    In Problem \ref{problem:matdot}, if $\textbf{d}$ is zero in some coordinate $i$, that implies that both the elements of $D_A$ and $D_B$ are zero in $i$. So we can project over $[l] \setminus \{i\}$ and obtain a lower recovery threshold algorithm since it uses fewer variables but maintains the footprint.
\end{remark}

We proceed studying different solutions to Problem \ref{problem:matdot} and their recovery thresholds.

\subsection{The box and the case $q=2$}

The simplest construction we can come up with is, as in Section \ref{section:polynomial}, choosing a ``box'' for $D_A$ and $D_B$. This is Construction \ref{construction:box_matdot}, which gives a simple way of defining $D_A$ and $D_B$ while generalizing the optimal solution for classical matdot codes.

\begin{construction}[Box] \label{construction:box_matdot}
    Let $\textbf{m} \in \mathbb{N}^\ell_{< q}$ such that $2(m_i - 1) < q$ for every $i = 1, 2, \ldots, l$. Define
    \begin{equation*}
        D_A :=  D_B = \{\textbf{a} \in \mathbb{N}^l_{<q} \such a_i < m_i\}.
    \end{equation*}
    Because of the election of each $m_i$, we have that $\textbf{a} + \textbf{b} \in \mathbb{N}^l_{<q}$ for every $\textbf{a} \in D_A$ and $\textbf{b} \in D_B$, resulting in $D_A +_q D_B = D_A + D_B$ being the classical Minkowski sum. The rest of the definition ensures $D_A$ and $D_B$ to be a solution to Problem \ref{problem:matdot}, with $\textbf{d} = (m_1 - 1, m_2 - 1, \ldots, m_l - 1)$. The footprint results to be
    \begin{equation*}
        \FB(D_A + D_B) = \prod_{i=1}^l (q - 2m_i + 2).
    \end{equation*}
\end{construction}

In Construction \ref{construction:box_matdot} we have the restriction $2(m_i - 1) < q$ which makes it non applicable for lower values of $q$. The analog in Section \ref{section:polynomial} was Construction \ref{construction:box_poly}, which suffered from the same limitation, in particular for $q=2$. We were capable of tackling this issue in Section \ref{section:polynomial} by separating variables (see Construction \ref{construction:separation_of_variables}), but in the case of multivariate matdot codes, using $q=2$ is intractable as Proposition \ref{proposition:no_matdot2} shows.

\begin{notation}
    Let $S \subseteq \mathbb{N}^l$, we denote $\supp(S) := \{i \in [l] \such \exists \textbf{s} \in S \quad s_i \neq 0 \}$.
\end{notation}

\begin{proposition} \label{proposition:no_matdot2}
    Let $D_A$, $D_B$ and $\textbf{d}$ be a solution for Problem \ref{problem:matdot} with $q=2$. Then $\FB(D_A +_q D_B) = 2^{l - |\supp(D_A +_q D_B)|}$. In particular, if $\supp(D_A +_q D_B) = [l]$ (see Remark \ref{remark:useless_variables} for its importance), then $\FB(D_A +_q D_B) = 1$.
\end{proposition}

\begin{proof}
    Since $d_i = (a_i + b_i)_q$ for some $\textbf{a} \in D_A$ and $\textbf{b} \in D_B$, then $d_i = 1$ if $a_i = 1$ or $b_i = 1$, otherwise $d_i = 0$. So $\FB(D_A +_q D_B) = \prod_{i=1}^l (2 - d_i) = 2^{l - |\supp(D_A +_q D_B)|}$.
\end{proof}

As stated in Proposition \ref{proposition:no_matdot2}, there are no methods for DMM with straggler tolerance for matrices over $\mathbb{F}_2$ using multivariate matdot codes (apart from trivial ones). So Construction \ref{construction:separation_of_variables} in Subsection \ref{subsection:few_variables} still remains the best option for performing DMM over $\mathbb{F}_2$.

Let us study a different way of picking $D_A$ and $D_B$.

\subsection{Optimal solution when $D_A = D_B$}

In Construction \ref{construction:box_matdot}, we picked $D_A$ and $D_B$ as the same set. This is not a requirement to fullfill but both in matdot codes \cite{polydot_codes} and in AG matdot codes \cite{fidalgo2024distributed}, optimality is achieved when doing so. This motivates exploring the idea of restricting to $D_A = D_B$, where the optimality can be satisfactorily studied.

\begin{definition}
    We say $D \subseteq \mathbb{R}^l$ is convex if, for every $\textbf{a}, \textbf{a}^\prime \in D$ and $\lambda \in [0,1] \subseteq \mathbb{R}$, then $\lambda \textbf{a} + (1 - \lambda) \textbf{a}^\prime \in D$.
\end{definition}

The following result is already known but we do not know any explicit proof in the literature.

\begin{lemma} \label{lemma:hyperbolic_convex}
    Consider $\Hyp_q(F,l)^* := \{\textbf{a} \in \mathbb{R}^l \such \prod_{i=1}^l (q - a_i) \geq F\}$. Then, $\Hyp_q(F,l)^*$ is convex.
\end{lemma}

\begin{proof}
    Let $\textbf{a}, \textbf{a}^\prime \in \Hyp_q(F,l)^*$ and $\lambda \in [0,1] \subseteq \mathbb{R}$. Then
    \begin{equation*}
            \prod_{i=1}^l (q - \lambda a_i - (1 - \lambda) a_i^\prime) = \prod_{i=1}^l (\lambda (q - a_i) + (1 - \lambda) (q - a_i^\prime) )
    \end{equation*}
    Using logarithms, which are known to satisfy that
    \begin{equation*}
        \log (\lambda x + (1 - \lambda) y) \geq \lambda \log (x) + (1 - \lambda) \log (y),
    \end{equation*}
    for $x, y > 0$, we obtain:
    \begin{equation*}
        \begin{split}
            \log \left (\prod_{i=1}^l (\lambda (q - a_i) + (1 - \lambda) (q - a_i^\prime) ) \right ) &= \sum_{i=1}^l \log (\lambda (q - a_i) + (1 - \lambda) (q - a_i^\prime) )\\
            &\geq \sum_{i=1}^l (\lambda \log (q - a_i) + (1 - \lambda) \log (q - a_i^\prime))\\
            &=\lambda \sum_{i=1}^l \log (q - a_i) + (1 - \lambda) \sum_{i=1}^l \log (q - a_i^\prime)\\
            &=\lambda \log \prod_{i=1}^l (q - a_i) + (1 - \lambda) \log \prod_{i=1}^l (q - a_i^\prime)\\
            &\geq \lambda \log(F) + (1 - \lambda) \log(F)\\
            &=\log(F).
        \end{split}
    \end{equation*}
    Finally, because the logarithm is an increasing function, we conclude that
    \begin{equation*}
        \prod_{i=1}^l (q - \lambda a_i - (1 - \lambda) a_i^\prime) \geq F.
    \end{equation*}
\end{proof}

Now, we state Proposition \ref{proposition:convex_ruano} which also can be found in \cite[Prop. 4]{garcia2020high}. Since the proof is short, we give it for completion.

\begin{proposition}[{\cite[Prop. 4]{garcia2020high}}] \label{proposition:convex_ruano}
    Let $D \subseteq \mathbb{N}^l_{<\frac{q}{2}}$. Then $\FB(D + D) \geq F$ if and only if $\FB(2 \textbf{a}) \geq F$ for every $\textbf{a} \in D$.
\end{proposition}

\begin{proof}
    First, if $\FB(D + D) \geq F$, then $\FB(\textbf{a} + \textbf{a}) \geq F$ holds as a particular case. Conversely, suppose that, for every $\textbf{a} \in D$, $\FB(2 \textbf{a}) \geq F$ holds, which is equivalent to saying that $2 \textbf{a} \in \Hyp_q(F,l)^*$. Let $\textbf{a}, \textbf{a}^\prime \in D$ and consider $\textbf{a} + \textbf{a}^\prime$. Since $2 \textbf{a}, 2 \textbf{a}^\prime \in \Hyp_q(F, l)^*$, we apply Lemma \ref{lemma:hyperbolic_convex} to conclude that $\textbf{a} + \textbf{a}^\prime = \frac{1}{2} (2 \textbf{a}) + \frac{1}{2} (2 \textbf{a}^\prime) \in \Hyp_q(F, l)^*$, that is $\FB(\textbf{a} + \textbf{a}^\prime) \geq F$.
\end{proof}

Proposition \ref{proposition:convex_ruano} allows us to simplify the computation of $\FB(D_A + D_B)$ when $D_A = D_B$. We exploit this fact and propose Construction \ref{construction:halfhyperbolic}.

\begin{construction}[Half hyperbolic] \label{construction:halfhyperbolic}
    Let $F \in \mathbb{N}$ and $\textbf{d} \in \mathbb{N}_{<\frac{q}{2}}^l$. Define
    \begin{equation*}
        D_A := D_B = \{\textbf{a} \in \mathbb{N}_{< \frac{q}{2}}^l \such \FB(2\textbf{a}) \geq F, \quad \FB(2(\textbf{d} -\textbf{a})) \geq F\}.
    \end{equation*}
    Since $\textbf{a} \in \mathbb{N}_{< \frac{q}{2}}^l$, $D_A +_q D_B$ is the usual Minkowski sum. We conclude from $\textbf{a} \in D_A \iff \textbf{d} - \textbf{a} \in D_B$ that $D_A$ and $D_B$ are a solution to Problem \ref{problem:matdot}. Moreover, by Proposition \ref{proposition:convex_ruano} we obtain that
    $\FB(D_A + D_B) \geq F$.
\end{construction}

\begin{figure}
    \centering
    \resizebox{!}{200pt}{
        \begin{tikzpicture}
	\draw[thick,->] (0,0) -- (15.6, 0) node {};
	\draw[thick,->] (0,0) -- (0, 15.6) node {};
	\draw (0 cm, 0pt) node[anchor=north east] {$0$};
	\foreach \x in {5,10,...,15}
	    \draw (\x cm, -1pt) node[anchor=north] {$\x$};
	\foreach \y in {5,10,...,15}
	    \draw (-1pt, \y cm) node[anchor=east] {$\y$};

	\foreach \x in {0,...,15}
	    \foreach \y in {0,...,15}
	        \fill[color=gray] (\x, \y) circle[radius=1pt] node {};

	\fill[fill= blue, opacity=0.5](0, 5) -- (0, 2) -- (2, 0) -- (5, 0) -- (5, 3) -- (3, 5) -- cycle; 
	\fill[fill= red, opacity=0.5](0, 10) -- (0, 4) -- (4, 0) -- (10, 0) -- (10, 6) -- (6, 10) -- cycle; 
\end{tikzpicture}
    }
    \caption{}
\end{figure}

\begin{remark}
    Observe that Construction \ref{construction:halfhyperbolic} generalizes Construction \ref{construction:box_matdot} by setting $\textbf{m} = \textbf{d}$ and $F = 0$.
\end{remark}

\begin{proposition} \label{proposition:matdot_equals_optimal}
    Let $D_A = D_B$ and $\textbf{d}$ be a solution to Problem \ref{problem:matdot} such that $D_A +_q D_B = D_A + D_B$. Then there exist a solution $D_A^\prime = D_B^\prime$ and $\textbf{d}$ given by Construction \ref{construction:halfhyperbolic} such that $|D_A| \leq |D_A^\prime|$ and $\FB(D_A + D_B) \leq \FB(D_A^\prime + D_B^\prime)$. Consequently, if we choose $D \subseteq D_A^\prime$ of size $|D| = |D_A|$, then $D$ and $\textbf{d} - D$ is a solution of size $|D_A|$ with $\FB(D_A + D_B) \leq \FB(D + (\textbf{d} - D))$, thus with smaller or equal recovery threshold.
\end{proposition}

\begin{proof}
    Let $F := \FB(D_A + D_B)$. For every $\textbf{a} \in D_A$ it holds that $\FB(2\textbf{a}) \geq F$. Consider $D_A^\prime$ and $D_B^\prime$ given by Construction \ref{construction:halfhyperbolic} with parameters $F$ and $\textbf{d}$. Because of Proposition \ref{proposition:convex_ruano}, $\FB(D_A^\prime + D_B^\prime) \geq F$ and the statement of the proposition holds.
\end{proof}

Proposition \ref{proposition:matdot_equals_optimal} shows that Construction \ref{construction:halfhyperbolic} yields the best solutions among the ones of the form $D_A = D_B$. So solutions satisfying $D_A = D_B$ can be seen as two subsets of the ones of Construction \ref{construction:halfhyperbolic}. This easily motivates Construction \ref{construction:halfhyperbolic}.

Similarly to Proposition \ref{proposition:recursive_better_box}, the next proposition shows how to compute the size of
$D_A$ and $D_B$ recursively in Construction \ref{construction:halfhyperbolic}. In fact, the same recurrence show us how to compute the sets itselves explicitly. We give only the proof for the sizes for the sake of brevity.

\begin{proposition}
    Let $F, G \in \mathbb{N}$ and $\textbf{d} \in \mathbb{N}_{<q}^l$. Consider
    \begin{equation*}
        D(F, G, l) := \left \{\textbf{a} \in \mathbb{N}_{<\frac{q}{2}}^l \such \prod_{i=1}^l (q - 2a_i) \geq F, \quad \prod_{i=1}^l (q - 2 d_i + 2 a_i) \geq G \right \}.
    \end{equation*}
    In particular, $D(F, F, l) = D_A = D_B$ in Construction \ref{construction:halfhyperbolic}. Then
    \begin{equation*}
        |D(F, G, l)| =
        \begin{cases}
            \max \{0, \lfloor \frac{q - F}{2} \rfloor - \lceil \frac{G - q + 2 d_1}{2} \rceil + 1\} & \text{if } l = 1\\
            \sum_{a=1}^{\left \lfloor \frac{q - 1}{2} \right \rfloor} \left | D \left (\left \lceil \frac{F}{q - 2 a} \right \rceil, \left \lceil \frac{G}{q - 2 d_l +  2 a } \right \rceil, l-1 \right) \right| & \text{if } l > 1
        \end{cases}
    \end{equation*}
\end{proposition}

\begin{proof}
    If $l = 1$, then
    \begin{equation*}
        \begin{split}
            |D(F, G, 1)|
            &= \left | \left \{ \textbf{a} \in \mathbb{N}_{< \frac{q}{2}}^l \such q - 2a_1 \geq F, \quad q - 2 d_1 + 2a_1 \geq G \right \} \right | =\\
            &= \left |\left \{ \textbf{a} \in \mathbb{N}_{< \frac{q}{2}}^l \such \left \lfloor \frac{q - F}{2} \right \rfloor \geq a_1 \geq \left \lceil \frac{G - q + 2d_1}{2} \right \rceil \right \} \right | =\\
            &=\max \left \{0, \left \lfloor \frac{q - F}{2} \right \rfloor - \left \lceil \frac{G - q + 2d_1}{2} \right \rceil + 1 \right \}.
        \end{split}
    \end{equation*}
    If $l > 1$, we partition $D(F,G,l)$ depending on the last coordinate of each element and the recursive formula follows easily:
    \begin{equation*}
        \begin{split}
            |D(F, G, l)|
            &= \Biggl | \bigcup_{a=0}^{\left \lfloor \frac{q - 1}{2} \right \rfloor} \Biggl \{(a_1, a_2, \ldots, a_{l-1}, a) \in \mathbb{N}_{<\frac{q}{2}}^l \such \prod_{i=1}^{l-1} (q - 2a_i) \geq \left \lceil \frac{F}{q - 2a} \right \rceil \\
            &\quad \,\,\,\, \prod_{i=1}^{l-1} (q - 2d_i + 2 a_i) \geq \left \lceil \frac{G}{q - 2d_l + 2a}\right \rceil \Biggr \} \Biggr |\\
            &= \sum_{a=0}^{\left \lfloor \frac{q - 1}{2} \right \rfloor} \Biggl | \Biggl \{(a_1, a_2, \ldots, a_{l-1}, a) \in \mathbb{N}_{<\frac{q}{2}}^l \such \prod_{i=1}^{l-1} (q - 2a_i) \geq \left \lceil \frac{F}{q - 2a} \right \rceil \\
            &\quad \,\,\,\, \prod_{i=1}^{l-1} (q - 2 d_i + 2 a_i) \geq \left \lceil \frac{G}{q - 2 d_l + 2a}\right \rceil \Biggr \} \Biggr |\\
            &= \sum_{a=1}^{\left \lfloor \frac{q - 1}{2} \right \rfloor} \left | D \left (\left \lceil \frac{F}{q - 2 a} \right \rceil, \left \lceil \frac{G}{q - 2 d_l +  2 a } \right \rceil, l-1 \right) \right|.
        \end{split}
    \end{equation*}
\end{proof}

\section*{Appendices}
\appendix

\section{Complexity} \label{appendix:complexity}

The complexity analysis is exactly the same as in \cite[Sec. 5]{fidalgo2024distributed}. See that reference for a more detailed treatment of the complexity. As in the rest of the manuscript, recall that $k+1$ denotes the recovery threshold.

In multivariate polynomial codes (respectively, multivariate matdot codes), each worker node has to multiply two matrices of sizes $\frac{r}{m} \times s$ and $s \times \frac{t}{n}$ (resp. $r \times \frac{s}{m}$ and $\frac{s}{m} \times t$). Using the naive matrix multiplication algorithm, this has complexity $\mathcal{O}(\frac{rt}{mn})$ (resp. $\mathcal{O}(\frac{rst}{m})$).

For the decoding process, in multivariate polynomial codes we have to interpolate each entry of $h \in \mathcal{R}^{\frac{r}{m} \times \frac{t}{n}}$. That is, interpolating a polynomial of degree at most $k$, which has complexity $\mathcal{O}(k^2)$ entrywise. In total, $\mathcal{O}(\frac{rt}{mn}k^2)$. When considering multivariate matdot codes, we only have to interpolate one coordinate of $h$, having this complexity $\mathcal{O}(rtk)$. For both multivariate polynomial and matdot codes, we have to add the complexity of inverting the matrix which gives the solutions to the linear system of equations for interpolating. This results in complexities $\mathcal{O}(\frac{r t}{mn} k^2 + k^3)$ and $\mathcal{O}(rtk+k^3)$, respectively.

We summarize the complexities in Table \ref{table:complexity}.

\begin{table}[h] \label{table:complexity}
    \centering
    \begin{tabular}{l|c|c}
                            & Worker computation              & Decoding computation \\ \hline
        Multivariate polynomial codes & $\mathcal{O}(\frac{r s t}{mn})$ & $\mathcal{O}(\frac{rt}{mn} k^2 +k^3)$ \\ \hline
        Multivariate matdot codes     & $\mathcal{O}(\frac{r s t}{m})$  & $\mathcal{O}(rtk+k^3)$   \\ \hline
    \end{tabular}
    \caption{Complexity of multivariate polynomial and matdot codes.}
    \label{table:complexity}
\end{table}

\section{Tables}

\subsection{Multivariate polynomial codes}

We present some tables concerning the constructions of Section \ref{section:polynomial}. Tables \ref{table:poly_box1} and \ref{table:poly_box2} contain parameters for both Constructions \ref{construction:box_poly} and \ref{construction:better_box}. Tables \ref{table:poly_varsep1} and \ref{table:poly_varsep2}, for Construction \ref{construction:separation_of_variables}.

\begin{table}[H]
\centering
\caption{Given $q = 19, l = 2$, Construction \ref{construction:box_poly} with $n_i = \lfloor \frac{2q}{3m_i} \rfloor$ and Construction \ref{construction:better_box} with $|D_B| = \tilde{n}$. The bound of Proposition \ref{proposition:polynomial_bound} for each construction are $\xi$ and $\tilde{\xi}$, respectively. The number of worker nodes is $N = 361$ and the recovery threshold, $k+1$.}
\begin{tabular}{|c|c|c|c:c|c|c:c|c|}
\hline
$m_i$ & $n_i$ & $m$ & $n$ & $\tilde{n}$ & $\FB(D_A + D_B)$ & $\xi$ & $\tilde{\xi}$ & $k+1$ \\
\hline
$1$ & $12$ & $1$ & $144$ & $204$ & $64$ & $102$ & $64$ & $298$ \\
$2$ & $6$ & $4$ & $36$ & $48$ & $64$ & $102$ & $70$ & $298$ \\
$3$ & $4$ & $9$ & $16$ & $20$ & $64$ & $102$ & $77$ & $298$ \\
$4$ & $3$ & $16$ & $9$ & $11$ & $64$ & $102$ & $80$ & $298$ \\
$5$ & $2$ & $25$ & $4$ & $4$ & $100$ & $143$ & $143$ & $262$ \\
$6$ & $2$ & $36$ & $4$ & $4$ & $64$ & $102$ & $102$ & $298$ \\
\hline
\end{tabular}
\label{table:poly_box1}
\end{table}

\begin{table}[H]
\centering
\caption{Given $q = 25, l = 2$, Construction \ref{construction:box_poly} with $n_i = \lfloor \frac{2q}{3m_i} \rfloor$ and Construction \ref{construction:better_box} with $|D_B| = \tilde{n}$. The bound of Proposition \ref{proposition:polynomial_bound} for each construction are $\xi$ and $\tilde{\xi}$, respectively. The number of worker nodes is $N = 625$ and the recovery threshold, $k+1$.}
\begin{tabular}{|c|c|c|c:c|c|c:c|c|}
\hline
$m_i$ & $n_i$ & $m$ & $n$ & $\tilde{n}$ & $\FB(D_A + D_B)$ & $\xi$ & $\tilde{\xi}$ & $k+1$ \\
\hline
$1$ & $16$ & $1$ & $256$ & $364$ & $100$ & $168$ & $100$ & $526$ \\
$2$ & $8$ & $4$ & $64$ & $84$ & $100$ & $168$ & $115$ & $526$ \\
$3$ & $5$ & $9$ & $25$ & $31$ & $121$ & $192$ & $152$ & $505$ \\
$4$ & $4$ & $16$ & $16$ & $20$ & $100$ & $168$ & $126$ & $526$ \\
$5$ & $3$ & $25$ & $9$ & $11$ & $121$ & $192$ & $154$ & $505$ \\
$6$ & $2$ & $36$ & $4$ & $4$ & $196$ & $270$ & $270$ & $430$ \\
\hline
\end{tabular}
\label{table:poly_box2}
\end{table}

\begin{table}[H]
\centering
\caption{Given $q = 2, l = 10$, Construction \ref{construction:separation_of_variables} with $m^\prime = 5$, $n^\prime = 5$ and $F_A = F_B = F$. The bound of Proposition \ref{proposition:polynomial_bound} is $\xi$. The number of worker nodes is $N = 1024$ and the recovery threshold, $k+1$.}
\begin{tabular}{|c|c|c|c|c|c|}
\hline
$F$ & $m$ & $\FB(D_A + D_B)$ & $\xi$ & $k+1$ \\
\hline
$2$ & $31$ & $4$ & $8$ & $1021$ \\
$4$ & $26$ & $16$ & $16$ & $1009$ \\
$8$ & $16$ & $64$ & $64$ & $961$ \\
$16$ & $6$ & $256$ & $256$ & $769$ \\
\hline
\end{tabular}
    \label{table:poly_varsep1}
\end{table}

\begin{table}[H]
\centering
\caption{Given $q = 2, l = 20$, Construction \ref{construction:separation_of_variables} with $m^\prime = 10$, $n^\prime = 10$ and $F_A = F_B = F$. The bound of Proposition \ref{proposition:polynomial_bound} is $\xi$. The number of worker nodes is $N = 1048576$ and the recovery threshold, $k+1$.}
\begin{tabular}{|c|c|c|c|c|c|}
\hline
$F$ & $m$ & $\FB(D_A + D_B)$ & $\xi$ & $k+1$ \\
\hline
$2$ & $1023$ & $4$ & $16$ & $1048573$ \\
$4$ & $1013$ & $16$ & $64$ & $1048561$ \\
$8$ & $968$ & $64$ & $128$ & $1048513$ \\
$16$ & $848$ & $256$ & $512$ & $1048321$ \\
$32$ & $638$ & $1024$ & $2048$ & $1047553$ \\
$64$ & $386$ & $4096$ & $4096$ & $1044481$ \\
$128$ & $176$ & $16384$ & $16384$ & $1032193$ \\
$256$ & $56$ & $65536$ & $65536$ & $983041$ \\
$512$ & $11$ & $262144$ & $262144$ & $786433$ \\
\hline
\end{tabular}
    \label{table:poly_varsep2}
\end{table}

\begin{table}[H]
\centering
\caption{Given $q = 64, l = 2$, Construction \ref{construction:separation_of_variables} with $m^\prime = 1$, $n^\prime = 1$ and $F_A = F_B = F$. The bound of Proposition \ref{proposition:polynomial_bound} is $\xi$. The number of worker nodes is $N = 4096$ and the recovery threshold, $k+1$.}
\begin{tabular}{|c|c|c|c|c|c|}
\hline
$F$ & $m$ & $\FB(D_A + D_B)$ & $\xi$ & $k+1$ \\
\hline
$2$ & $63$ & $4$ & $35$ & $4093$ \\
$4$ & $61$ & $16$ & $92$ & $4081$ \\
$8$ & $57$ & $64$ & $236$ & $4033$ \\
$16$ & $49$ & $256$ & $600$ & $3841$ \\
$32$ & $33$ & $1024$ & $1540$ & $3073$ \\
\hline
\end{tabular}
    \label{table:poly_varsep3}
\end{table}

\begin{table}[H]
\centering
\caption{Given $q = 128, l = 2$, Construction \ref{construction:separation_of_variables} with $m^\prime = 1$, $n^\prime = 1$ and $F_A = F_B = F$. The bound of Proposition \ref{proposition:polynomial_bound} is $\xi$. The number of worker nodes is $N = 16384$ and the recovery threshold, $k+1$.}
\begin{tabular}{|c|c|c|c|c|c|}
\hline
$F$ & $m$ & $\FB(D_A + D_B)$ & $\xi$ & $k+1$ \\
\hline
$2$ & $127$ & $4$ & $60$ & $16381$ \\
$4$ & $125$ & $16$ & $156$ & $16369$ \\
$8$ & $121$ & $64$ & $396$ & $16321$ \\
$16$ & $113$ & $256$ & $980$ & $16129$ \\
$32$ & $97$ & $1024$ & $2440$ & $15361$ \\
$64$ & $65$ & $4096$ & $6215$ & $12289$ \\
\hline
\end{tabular}
    \label{table:poly_varsep4}
\end{table}

\subsection{Multivariate matdot codes}

\begin{table}[H]
\centering
\caption{Given $q = 8, l = 3$, Construction \ref{construction:halfhyperbolic} with $d$ computed to maximize $m$ given the footprint $F$. The number of worker nodes is $N = 512$ and the recovery threshold, $k+1$.}
\begin{tabular}{|c|c|c|c|}
\hline
$F$ & $m$ & $k+1$ \\
\hline
$1$ & $64$ & $512$ \\
$9$ & $62$ & $504$ \\
$17$ & $56$ & $496$ \\
$25$ & $50$ & $488$ \\
$33$ & $38$ & $480$ \\
$41$ & $38$ & $472$ \\
$49$ & $26$ & $464$ \\
$57$ & $26$ & $456$ \\
\hline
\end{tabular}
    \label{table:matdot1}
\end{table}

\begin{table}[H]
\centering
\caption{Given $q = 32, l = 3$, Construction \ref{construction:halfhyperbolic} with $d$ computed to maximize $m$ given the footprint $F$. The number of worker nodes is $N = 32768$ and the recovery threshold, $k+1$.}
\begin{tabular}{|c|c|c|c|}
\hline
$F$ & $m$ & $k+1$ \\
\hline
$1$ & $4096$ & $32768$ \\
$513$ & $3044$ & $32256$ \\
$1025$ & $2106$ & $31744$ \\
$1537$ & $1440$ & $31232$ \\
$2049$ & $976$ & $30720$ \\
$2561$ & $622$ & $30208$ \\
$3073$ & $374$ & $29696$ \\
$3585$ & $176$ & $29184$ \\
\hline
\end{tabular}
    \label{table:matdot2}
\end{table}

\printbibliography

\end{document}